\documentclass[10pt,conference]{IEEEtran}

\usepackage{amssymb}
\usepackage{amsmath}
\usepackage{amsmath,bm}
\usepackage{amsthm}
\usepackage{graphicx}
\usepackage{subfigure}
\usepackage{cite}
\usepackage{enumerate}
\usepackage{url}
\usepackage{bm}
\usepackage[ruled]{algorithm2e}
\usepackage{setspace}
\usepackage{color, soul}
\usepackage{epstopdf}
\usepackage[noend]{algpseudocode}
\usepackage{float}
\DeclareMathOperator*{\argmax}{argmax}
\usepackage{booktabs}
\usepackage{threeparttable}
\usepackage{longtable}
\usepackage{rotating}
\usepackage{multirow}
\usepackage{mathrsfs}
\usepackage{array}
\usepackage{stfloats}

\newtheorem{prop}{Proposition}

\def\BibTeX{{\rm B\kern-.05em{\sc i\kern-.025em b}\kern-.08em
    T\kern-.1667em\lower.7ex\hbox{E}\kern-.125emX}}
\begin{document}

\iffalse
\newtheorem{definition}{\bf~~Definition}
\newtheorem{theorem}{\bf~~Theorem}
\newtheorem{observation}{\bf~~Observation}
\newtheorem{proposition}{\bf~~Proposition}
\newtheorem{remark}{\bf~~Remark}
\newtheorem{lemma}{\bf~~Lemma}
\fi

\title{Near-Far Field Codebook Design for IOS-Aided Multi-User Communications}

\author{
\IEEEauthorblockN{Shupei Zhang\IEEEauthorrefmark{1}, Yutong Zhang\IEEEauthorrefmark{1}, Boya Di\IEEEauthorrefmark{1}}

\IEEEauthorblockA{\IEEEauthorrefmark{1}State Key Laboratory of Advanced Optical Communication Systems and Networks, Peking University, Beijing, China.}

\IEEEauthorblockA{Email: zhangshupei@pku.edu.cn, yutongzhang@pku.edu.cn, diboya@pku.edu.cn.}
}

\maketitle
\begin{abstract}
Recently, the rapid development of metasurface facilitates the growth of extremely large-scale antenna arrays, making the ultra-massive MIMO possible.
In this paper, we study the codebook design and beam training for an intelligent omni-surface~(IOS) aided multi-user system, where the IOS is a novel metasurface enabling simultaneous signal reflection and refraction.
To deal with the near field expansion caused by the large-dimension of IOS, we design a near-far field codebook to serve users both in the near and far fields without prior knowledge of user distribution.
Moreover, to fully exploit the dual functionality of the IOS, the coupling between the reflective and refractive signals is analyzed theoretically and utilized in the codebook design, thereby reducing the training overhead.
On this basis, the multi-user beam training is adopted where each codeword covers multiple areas to enable all users to be trained simultaneously. 
Simulation results verify our theoretical analysis on the reflective-refractive coupling.
Compared to the state-of-the-art schemes, the proposed scheme can improve the sum rate and throughput.
\end{abstract}

\vspace{-1.4em}
%%%%%%%%%%%%%%%%%%%%%%%%%%%%%%%%%%%%%%%%%%%%%%%
\section{Introduction\label{sec:intro}}%%%%%%%%
%%%%%%%%%%%%%%%%%%%%%%%%%%%%%%%%%%%%%%%%%%%%%%%
\vspace{-0.4em}
With the explosive growth in the number of mobile devices and the rapid development of emerging applications, future wireless communications expect a new technique to provide high-speed and seamless data services~\cite{6G}.
The intelligent omni-surface (IOS) is an emerging technique to enhance wireless communications that allows simultaneous signal reflection and refraction, thereby achieving full-dimensional wireless communications~\cite{IOS}.
However, numerous IOS elements lead to massive IOS aided links between BS and users, which makes the accurate channel state information (CSI) hard to acquire.

To deal with this problem, in the literature, existing works have designed beamforming schemes without perfect CSI via codebook design and beam training to match either the near-field~\cite{dai},~\cite{hierarchical} or far-field channel models~\cite{Jamali},~\cite{tradeoff}.
Authors in~\cite{dai} proposed the near-field beam training scheme for the reflective metasurface (i.e.,~RIS) aided system based on a hierarchical near-field codebook.
In~\cite{Jamali}, authors proposed a quadratic phase-shift design for the IOS which yields a higher power efficiency.

Nevertheless, the pure near-field and far-field codebooks in existing works may not be adaptive to the increasing number of IOS elements.
Specifically, due to the large dimension of the IOS, the boundary between the near field and the far field expands accordingly, and users are likely to distribute in both the near and far fields of the IOS~\cite{Rayleigh}.
The existing pure near-field and far-field codebooks will mismatch users without prior knowledge of user distribution, leading to severe performance loss.
Therefore, it is crucial to design a near-far field codebook applicable to all users.

In this paper, we consider a large-scale IOS aided multi-user system where the codebook and beam training scheme are designed for users in both near and far fields.
Unlike the traditional reflection-only RIS, there exist amplitude and phase couplings between the reflective and refractive coefficients of the IOS elements~\cite{phasecouple}.
Therefore, the reflective and refractive beams generated by an IOS phase configuration are dependent in terms of the beam directions.
Each codeword no longer corresponds to only one direction, but to a reflective direction and a refractive direction coupled to each other, which needs to be derived and considered in the codebook design for the IOS system.
This coupling can be exploited to reduce the number of codewords, which lowers the training overhead.
Meanwhile, this coupling splits the energy into reflective and refractive signals inherently, potentially reducing the signal power.
Therefore, it is essential to consider the coupling of IOS signals in beamforming.

New challenges have arisen in such a system.
$First$, due to the different properties of the near-field and far-field channels and unknown user distribution, it is non-trivial to design a codebook to be applied to all users in any location.
$Second$, unlike traditional reflective metasurface, the coupling between the reflective and refractive signals of IOS needs to be derived and considered in beamforming, introducing difficulty to the codebook design and beam training.
By addressing these challenges, we contribute to the state-of-the-art in the following~ways.
\begin{enumerate}
\item Considering an IOS aided system, we analyze the amplitude and phase coupling of the IOS elements theoretically, unveiling the symmetry relationship between the reflective and refractive signals of IOS in both near and far fields, which is applicable to any IOS structure.
\item Given the symmetry relationship between the reflective and refractive signals, we design a near-far field codebook to serve users in both near and far regions without prior knowledge of user distribution.
Benefiting from the symmetry relationship between the reflective and refractive signals, the codebook size is halved to reduce the beam training overhead.
Based on that, the multi-user beam training mechanism is adopted to serve multiple users~simultaneously.
\item Simulation results verify the theoretical analysis on the coupling between the reflective and refractive signals of IOS.
Moreover, the proposed scheme outperforms the state-of-the-art codebooks in terms of the sum rate and throughput and performs close to the perfect CSI case.
\end{enumerate}
\vspace{-1em}
%%%%%%%%%%%%%%%%%%%%%%%%%%%%%%%%%%%%%%%%%%%%%%%
\section{System Model\label{sec:model}}%%%%%%%%
%%%%%%%%%%%%%%%%%%%%%%%%%%%%%%%%%%%%%%%%%%%%%%%
\vspace{-0.4em}
In this section, an IOS-aided multi-user communication system is first introduced and modeled, and then the amplitude and phase coupling is described.
\vspace{-0.5em}
\subsection{Scenario Description}
\vspace{-0.3em}
As shown in Fig.~\ref{Fig:scenario}, we consider a downlink multi-user wireless system where a~$N_b$-element uniform planar array~(UPA) equipped base station~(BS) serves~$K$ $N_u$-antenna users.
We deploy an IOS to transmit the signals from the BS to users utilizing the dual functions of reflection and refraction.

However, with an extremely large dimension of the IOS, the near-field region expands accordingly.
Take the Rayleigh distance~$\frac{2d^2}{\lambda}$ as the boundary between the near and far fields of the IOS
\footnote{Here we use the Rayleigh distance to define a near-far field boundary which has been widely adopted.}
, where $d$ and $\lambda$ represent the largest dimension of the IOS and the wavelength of the signal, respectively~\cite{Rayleigh}.
For example, with a $0.5m\times0.5m$ IOS, the Rayleigh distance is more than 80m at 26 GHz.
Hence, users are likely to spread out in both the near and far fields of the IOS, where the former introduces extra electromagnetic response variations among different radiation elements, such that we model the IOS aided system by considering the different electromagnetic characteristics of the near and far fields.
\begin{figure}[t]
\setlength{\abovecaptionskip}{0pt}
\setlength{\belowcaptionskip}{0pt}
	\centering
    \includegraphics[width=0.38\textwidth]{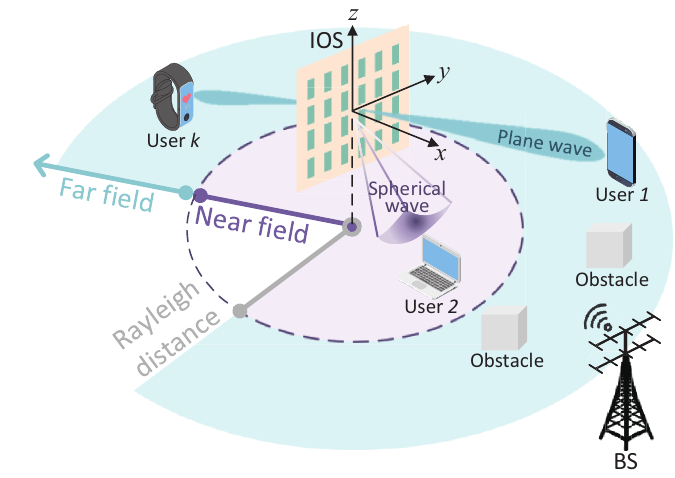}
	\caption{IOS aided multi-user communication system.}
	\label{Fig:scenario}
\vspace{0em}
\end{figure}
\vspace{0.6em}
\subsection{Coupling Between the Reflective-Refractive Beams at IOS}
\vspace{-0.2em}
We deploy an extremely large-scale IOS with~$L=L^h~\times~L^v$ elements between the BS and users to provide reflective/refractive links.
The reflective/refractive waves of IOS can be tuned toward the direction of users by manipulating the phase of each element.
The reflective and refractive signals may be coupled with each other, which is discussed as follows.

\subsubsection{Reflective-Refractive Amplitude Coupling}
Note that the energy of the incident signals is split by both reflective and refractive parts via the IOS. For convenience, we denote $\Gamma_t^l$ and $\Gamma_r^l$ as the amplitudes of the reflective and refractive signals with respect to the $l$-th element, satisfying that
\begin{equation}\label{amplitude}
\setlength{\abovedisplayskip}{0pt}
\setlength{\belowdisplayskip}{0pt}
(\Gamma_t^l)^2+(\Gamma_r^l)^2=1, 
\end{equation}
which implies the amplitude coupling between reflective and refractive signals.

\subsubsection{Reflective-Refractive Phase Coupling}
Define $\phi_t^l$ and $\phi_r^l$ as the reflection and refraction phase of the $l$-th element respectively.
We assume that the $l$-th IOS element can be encoded to perform $2^b$ potential phases to reflect/refract the radio wave, i.e., $\phi_t^l$ and $\phi_r^l$ $\in \{0,\frac{\pi}{2^{b-1}},...,\frac{(2^{b-1}-1)\pi}{2^{b-1}}\}$.
Due to the hardware limitations, for the same IOS element, $\phi_t^l$ and~$\phi_r^l$ are usually coupled with each other~\cite{phasecouple}
\begin{equation}\label{phase}
\setlength{\abovedisplayskip}{0pt}
\setlength{\belowdisplayskip}{0pt}
\phi_t^l-\phi_r^l=c,
\end{equation}
where $c$ is a constant related to the IOS structure.
Therefore, the reflection and refraction coefficients of the $l$-th element can be expressed by 
$\Psi_t^l=\Gamma_t^l e^{j\phi_t^l}, \Psi_r^l=\Gamma_r^l e^{j\phi_r^l}$.

\vspace{-1.2em}
\subsection{Channel Model}
\vspace{-0.5em}
\subsubsection{Channel between the BS and the IOS}
Considering the characteristic of spherical wave~\cite{nearfield}, the channel between the $l$-th element of IOS and the $n_b$-th antenna of the BS, i.e., the $(l,n_b)$-th element of the channel between the BS and the IOS~$\mathbf{H}_{BI}\in\mathbb{C}^{L\times N_b}$ can be expressed as~\cite{multibeam}
\begin{equation}\label{Hbr element}
\setlength{\abovedisplayskip}{0pt}
\setlength{\belowdisplayskip}{0pt}
[\mathbf{H}_{BI}]_{l,n_b}=\sqrt{\frac{1}{4\pi d_{l,n_b}^{2}}}e^{-j\frac{2\pi}{\lambda}d_{l,n_b}},
\end{equation}
where $d_{l,n_b}$ represents the distance from the $n_b$-th antenna to the $l$-th IOS element.
\subsubsection{Channel between the IOS and the user}
The electromagnetic radiation fields surrounding the IOS can be divided into far-field and near-field regions.
Denote the distance from IOS to user $k$ as $d_k$, then the user's channel model can be divided into two categories according to  $d_k$, i.e., the far-field and near-field channel model.

When user $k$ is in the near field, the waves between the IOS and the $k$-th user need to be accurately modeled according to the electromagnetic characteristic of \emph{spherical waves}~\cite{nearfield}.
The channel between the $l$-th element of IOS and the $n_u$-th antenna of user $k$ can be expressed as~\cite{multibeam}
\begin{equation}\label{Hbr element}
\setlength{\abovedisplayskip}{0pt}
\setlength{\belowdisplayskip}{0pt}
[\mathbf{H}_{IU,k}]_{n_u,l}=\sqrt{\frac{1}{4\pi d_{k,l,n_u}^{2}}}e^{-j\frac{2\pi}{\lambda}d_{k,l,n_u}},
\end{equation}
where $d_{k,l,n_u}$ represents the distance from the $l$-th IOS element to the $n_u$-th antenna of user $k$.

When user $k$ is in the far-field region, the waves between the IOS and the $k$-th user can be approximately modeled by \emph{planar waves}.
The channel between the $l$-th element of IOS and the $n_u$-th antenna of user $k$, i.e., the $(n_u,l)$-th element of the channel $\mathbf{H}_{IU,k}\in\mathbb{C}^{N_u\times L}$ can be expressed as~\cite{farcodebook}
\begin{equation}\label{HIU element}
\setlength{\abovedisplayskip}{0pt}
\setlength{\belowdisplayskip}{0pt}
[\mathbf{H}_{IU,k}]_{n_u,l}=\sqrt{\frac{1}{4\pi  d_k^2}}[\mathbf{u}(\theta_{k}^U,\psi_{k}^U)]_{n_u}[\mathbf{a}(\theta_{k}^I,\psi_{k}^I)]_{l},
\end{equation}
where $d_k$ represents the distance from the center of IOS to user~$k$ and $\theta_{k,j}^U$ represents the channel angle-of-arrival associated with user $k$.
The notations $\mathbf{u}(\theta_{k}^U,\psi_{k}^U)$ and $\mathbf{a}(\theta_{k}^I,\psi_{k}^I)$ represent the channel steering vectors associated with the user $k$ and IOS.
For example, $\mathbf{a}(\theta_{k}^I,\psi_{k}^I)$ is given by
\begin{equation}\label{Hbr element}
\begin{aligned}
\mathbf{a}(\theta_{k}^I,\psi_{k}^I)=&[1, e^{-j\frac{2\pi}{\lambda}\mu},...,e^{-j\frac{2\pi}{\lambda}(L^h-1)\mu}]^T\otimes \\&[1, e^{-j\frac{2\pi}{\lambda}\nu},...,e^{-j\frac{2\pi}{\lambda}(L^v-1)\nu}]^T,
\vspace{0em}
\end{aligned}
\end{equation}
where $\mu=d_h\cos\theta_{k}^I\sin\psi_{k}^I$ and $\nu=d_v\sin\theta_{k}^I$.
$\psi_{k}^I$ and~$\theta_{k}^I$ respectively represent the physical angles in the azimuth and
elevation at the IOS for channel $\mathbf{H}_{IU,k}$.
\vspace{-0.4em}
\subsection{Signal Model}
\vspace{-0.2em}
With the IOS based analog beamformer $\mathbf{\Psi}\in\mathbb{C}^{L\times L}$, the $k$-th user's received signal from the BS can be expressed as
\begin{equation}\label{signal}
\setlength{\abovedisplayskip}{0pt}
\setlength{\belowdisplayskip}{0pt}
y_k = \mathbf{w}_k^H\mathbf{H}_{IU,k}\mathbf{\Psi}\mathbf{H}_{BI}\mathbf{V}\mathbf{x}+\mathbf{w}_k^H\mathbf{n}_k,
\end{equation}
where $\mathbf{w}_k$ is the analog combiner at each user $k$. 
The vector $\mathbf{x}\in\mathbb{C}^{K\times 1}$ and $\mathbf{n}_k\sim\mathcal{C}\mathcal{N}(0,\sigma^2\mathbf{I}_{N_u})$ represent the transmit signals and noise, respectively.
The $(l,l)$-th element of the diagonal matrix $\mathbf{\Psi}\in\mathbb{C}^{L\times L}$ is $\Psi^l=\{\Psi_t^l, \Psi_r^l\}$.

The direction of the reflective or refractive signals is determined by the phase of each IOS element and the direction of the arrival signal from the BS concurrently. Considering such a coupling between the BS and the IOS, the digital precoder $\mathbf{V}$ and the IOS based analog beamformer $\mathbf{\Psi}$ should be designed jointly, expressed by
\begin{equation}\label{beamformer}
\setlength{\abovedisplayskip}{0pt}
\setlength{\belowdisplayskip}{0pt}
\mathbf{Q} = \mathbf{\Psi}\mathbf{H}_{BI}\mathbf{V},
\end{equation}
As such, the received signal in (\ref{signal}) can be rewritten as
\begin{equation}
\begin{aligned}
y_{k} & =\mathbf{w}_{k}^{H} \mathbf{H}_{I U, k} \mathbf{Q x}+\mathbf{w}_{k}^{H} \mathbf{n}_{k}, \\
& =\mathbf{w}_{k}^{H} \mathbf{H}_{I U, k}(\mathbf{q}_{k} x_{k}+\sum_{k^{\prime} \neq k} \mathbf{q}_{k^{\prime}} x_{k^{\prime}})+\mathbf{w}_{k}^{H} \mathbf{n}_{k},
\end{aligned}
\end{equation}
where $\mathbf{q}_{k} = \mathbf{\Psi}\mathbf{H}_{BI}\mathbf{v}_k$ and $\mathbf{v}_k$ are the $k$-th column of the IOS aided beamformer $\mathbf{Q}$ and the digital precoder $\mathbf{V}$, respectively.
\vspace{-1.6em}
%%%%%%%%%%%%%%%%%%%%%%%%%%%%%%%%%%%%%%%%%%%%%%%
\section{Problem Formulation\label{sec:scheme}}%%%%%%%%
%%%%%%%%%%%%%%%%%%%%%%%%%%%%%%%%%%%%%%%%%%%%%%%
\vspace{-0.1em}
In this section, we formulate the received power maximization problem and decompose it into codebook design, beam training, and beamforming subproblems.
\vspace{-0.4em}
\subsection{Problem Formulation}
To avoid acquiring perfect CSI, the codebook $\mathcal{Q}$=$[\mathbf{q}_1, \mathbf{q}_2, ...]$ consisting of both near/far-field codewords is set at the BS to generate the IOS aided beamformer $\mathbf{Q}$, while the codebook $\mathcal{W}$=$[\mathbf{w}_1, \mathbf{w}_2, ...]$ is set at the users to generate the analog combiner $\mathbf{w}_k$ via beam training.
Each codeword in $\mathcal{Q}$ can be used to align either a point in the near field or a direction in the far field, and the IOS aided beamformer $\mathbf{Q}$ designed by beam training can serve all users in both near and far fields simultaneously.

During the beam training, the BS broadcasts to all the users simultaneously using the codewords in $\mathcal{Q}$ and $\mathcal{W}$ sequentially.
After that, the optimal codewords $\mathbf{q}_{k}^*$ and $\mathbf{w}_{k}^*$ for user $k$ can be selected to design the IOS aided beamformer $\mathbf{Q}$ at the BS and the analog combiner at each user $k$, respectively, i.e., $\mathbf{q}_k = \mathbf{q}_{k}^*$ and $\mathbf{w}_k = \mathbf{w}_{k}^*$.
Finally, we use the IOS aided beamformer $\mathbf{Q}$ to configure the digital precoder $\mathbf{V}$ and the IOS based analog beamformer $\mathbf{\Psi}$.
The problem can be formulated as
\begin{subequations}\label{problem}
%\footnotesize
\setlength{\abovedisplayskip}{0pt}
\setlength{\belowdisplayskip}{2pt}
\begin{align}
\max\limits_{\mathbf{v}_k,\mathbf{w}_k, \mathbf{\Psi}}&\gamma_k,\\
\text{s.t.}&\gamma_k =
\begin{cases}
0, \mbox{if }\mathbf{c}_k \notin \mathscr{C}(\mathbf{q}_{k}^{*}) \cap \mathscr{C}(\mathbf{w}_{k}^{*}),\\
|(\mathbf{w}_{k}^{*})^{H} \mathbf{H}_{I U, k} \mathbf{q}_{k}^{*}|^{2}, \mbox{ otherwise},
\end{cases}\\
&\mathbf{\Psi} \mathbf{H}_{B I} \mathbf{v}_{k}=\mathbf{q}_{k}^{*} \in \mathcal{Q} \text { and } \mathbf{w}_{k}=\mathbf{w}_{k}^{*} \in \mathcal{W} ,\\
&\text{Tr}(\mathbf{V}^H\mathbf{V})\leq P_T, \\
&\eqref{amplitude} \text{ and } \eqref{phase}
\end{align}
\end{subequations}
where $\mathbf{c}_k$ represents the coordinates of user $k$.
Constraint~(10b) is the codebook design principle which implies the designed codebook should satisfy that the power gain of the received signal $\gamma_k$ is zero if the user is out of the coverage of the designed codewords, i.e., $\mathscr{C}(\mathbf{q}_{k}^{*})$ and $\mathscr{C}(\mathbf{w}_{k}^{*})$. Constraint (10c) that indicates we should configure the IOS based analog beamformer $\mathbf{\Psi}$ and the digital precoder $\mathbf{v}_k$ to approach the optimal codeword $\mathbf{q}_{k}^{*}$. Constraint (10d) is the power constraint for the BS, where $P_T$ is the transmit power budget available at the BS.
\vspace{-1em}
\subsection{Problem Decomposition}
We first design the codebooks $\mathcal{Q}$ and $\mathcal{W}$, and then select the optimal codewords $\mathbf{q}_k^{*}$ and $\mathbf{w}_k^{*}$ via beam training.
Given the beam training results, we configure the IOS based analog beamformer $\mathbf{\Psi}$ and the digital precoder $\mathbf{V}$ to approach the IOS aided beamformer $\mathbf{Q}$. Problem \eqref{problem} is thus decoupled into three subproblems below.

\subsubsection{Codebook Design Subproblem}
The coverage of the BS and IOS is equally discretized into $P$ areas and covered by codebook $\mathcal{Q}$, where the codewords can cover both the near and far regions of the IOS.
We now focus on the design of each codeword $\mathbf{q}$.
Denote all the indices of the areas in coverage of $\mathbf{q}$ by the collection $\mathscr{P(\mathbf{q})}$, the power gain at the $p$-th area with $\mathbf{q}$ should satisfy that
\begin{subequations}\label{sub1}
\setlength{\abovedisplayskip}{0pt}
\setlength{\belowdisplayskip}{2pt}
\begin{align}
\min\limits_{\mathbf{q} \in \mathcal{Q}} &\sum_{p=1}^{P} |\mathbf{h}_{p} \mathbf{q}-\eta|,\\
\text{s.t.}&\eta =
\begin{cases}
C, \mbox{if } p \in \mathscr{P}(\mathbf{q}),\\
0, \mbox{ otherwise}, 
\end{cases}
\end{align}
\end{subequations}
where $\mathbf{h}_{p}$ represents the equivalent channel between the IOS and the $p$-th area. It implies that the power gain at the $p$-th area can reach a constant $C$ if the $p$-th area is in the coverage of $\mathbf{q}$ and becomes $0$ if out of the coverage
\footnote{We employ a conventional far-field codebook at users as $\mathcal{W}$ whose details are not reiterated here~\cite{farcodebook}.}.
\subsubsection{Beam Training Subproblem}
By performing beam training, the optimal codewords $\mathbf{q}_{k}^{*}\in \mathcal{Q}$ and $\mathbf{w}_{k}^{*}\in \mathcal{W}$ for each user~$k$ that maximize the received power are obtained, expressed by
\begin{subequations}\label{sub2}
\setlength{\abovedisplayskip}{0pt}
\setlength{\belowdisplayskip}{2pt}
\begin{align}
\max\limits_{\mathbf{q}_{k}^{*}, \mathbf{w}_{k}^{*}} &\gamma_k,\\
\text{s.t.}&\mathbf{q}_{k}^{*}\in \mathcal{Q} \text { and } \mathbf{w}_{k}^{*}\in \mathcal{W}. 
\end{align}
\end{subequations}
\subsubsection{Beamforming Subproblem}
Given the beam training results, the IOS aided beamformer and analog combiner at each user~$k$ can be designed as $\mathbf{Q}$=$[\mathbf{q}_{1}^{*},...,\mathbf{q}_{K}^{*}]$ and $\mathbf{w}_{k}=\mathbf{w}_{k}^*$, respectively.
Then we use the IOS aided beamformer $\mathbf{Q}$ to configure the IOS based analog beamformer $\mathbf{\Psi}$ and the digital precoder $\mathbf{V}$ by solving the following subproblem, i.e.,
\begin{subequations}\label{sub3}
\setlength{\abovedisplayskip}{0pt}
\setlength{\belowdisplayskip}{2pt}
\begin{align}
\min\limits_{\mathbf{v}_k,\mathbf{\Psi}} &f\left(\mathbf{v}_{k}, \mathbf{\Psi}\right)=\sum_{k=1}^{K}\left|\mathbf{q}_{k}^{*}-\mathbf{\Psi} \mathbf{H}_{B I} \mathbf{v}_{k}\right|^{2},\\
\text{s.t.}&\text{Tr}(\mathbf{V}^H\mathbf{V})\leq P_T,\\
&\eqref{amplitude} \text{ and } \eqref{phase}
\end{align}
\end{subequations}

\begin{figure}[t]
\setlength{\abovecaptionskip}{0pt}
\setlength{\belowcaptionskip}{0pt}
	\centering
    \includegraphics[width=0.38\textwidth]{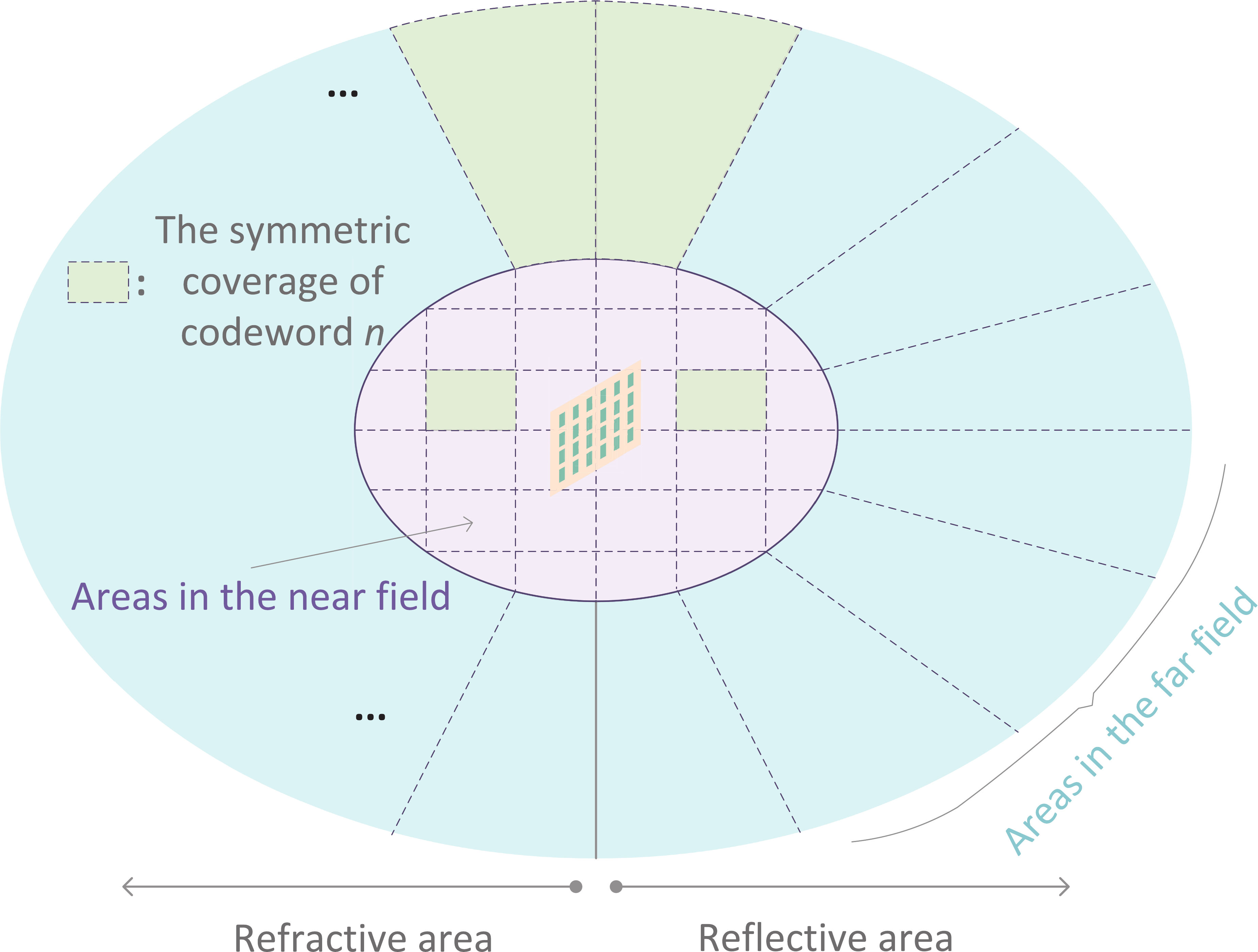}
	\caption{The coverage of each codeword in the near-far field codebook.}
	\label{Fig:codeword}
\vspace{0em}
\end{figure}
\vspace{-1em}
\section{Near-Far Field Codebook and Multi-User Beam Training Mechanism Design}\label{Codebook}
\vspace{-0.5em}
In this section, we first design a near-far field codebook to solve the subproblem~\eqref{sub1}.
To reduce the training overhead, the multi-user beam training mechanism is then designed to cope with the subproblem~\eqref{sub2} which can be performed for all users simultaneously.
Finally, given the training results, we configure the IOS based analog beamformer $\mathbf{\Psi}$ and the digital precoder $\mathbf{V}$ by solving the subproblem (13).
\vspace{-0.6em}
\subsection{Near-Far Field Codebook Design}
\vspace{-0.4em}
As shown in Fig.~\ref{Fig:codeword}, to serve users in both near and far regions, we design a near-far field codebook where each codeword aligns either a point in the near field or a direction in the far field.
The coverage of the BS and IOS is equally discretized into $P$ areas, and the equivalent channel between the IOS and the $p$-th area in the far field and that in the near field can be expressed~as

\vspace{-1em}
\begin{small}
\begin{equation}\label{channelp}
\mathbf{h}_p=
\begin{cases}
\sqrt{\frac{1}{4\pi d_p^2}}e^{-j\frac{2\pi}{\lambda}}\mathbf{a}^T(\theta_p,\psi_p), \mbox{if }d_p \geq \frac{2d^2}{\lambda},\\
[\sqrt{\frac{1}{4\pi d_{p,1}^{2}}}e^{-j\frac{2\pi}{\lambda}d_{p,1}},...,\sqrt{\frac{1}{4\pi d_{p,L}^{2}}}e^{-j\frac{2\pi}{\lambda}d_{p,L}}]^T,\mbox{ otherwise},
\end{cases}\\
\end{equation}
\end{small}where $d_p$ is the distance between the $p$-th area and center of the IOS.
$d_{p,l}$ represents the distance from the $l$-th element of the IOS to the $p$-th area.

Let the vector $\mathbf{u}\in\mathbb{C}^{P\times 1}$ indicate whether the $p$-th area is in the set $\mathscr{P}(\mathbf{q})$, the codeword $\mathbf{q}$ can be calculated by 
\vspace{-0.5em}
\begin{equation}\label{codeword}
\mathbf{q}=C \hat{\mathbf{H}}^{H}\left(\hat{\mathbf{H}} \hat{\mathbf{H}}^{H}\right)^{-1} \mathbf{u}
\vspace{-0.5em}
\end{equation}
where $\hat{\mathbf{H}}=[\mathbf{h}_1,...,\mathbf{h}_P]^T$.

According to the reflective-refractive phase coupling in~\eqref{phase}, we derive \emph{Proposition}~\ref{analysis} to show the relationship between the reflective and refractive signals.
Based on \emph{Proposition}~\ref{analysis}, each codeword simultaneously covers the reflective areas $\{p_{t_1},...,p_{t_n}\}$ and the refractive areas $\{p_{r_1},...,p_{r_n}\}$, where $\{p_{t_1},...,p_{t_n}\}$ and $\{p_{r_1},...,p_{r_n}\}$ are symmetrical about the IOS.
Considering such a reflective-refractive coupling, the codebook size can be halved, reducing the beam training overhead.
\begin{prop}\label{analysis}
Due to the dual functionality of signal reflection and refraction in~\eqref{phase}, when the IOS reflects the incident signal towards the direction $\phi_t$ in the far field or the point $(x_t,y_t,z_t)$ in the near field, a symmetrical refractive beam is simultaneously constructed which aligns with the direction $\phi_r=-\phi_t$ in the far field or the point $(x_r,y_r,z_r)=(x_t,-y_t,z_t)$ in the near field.
\end{prop}

\begin{proof}
Given the IOS aided beamformer of user $k$ according to the reflective beams $\mathbf{q}_k^t=\mathbf{\Psi}_t \mathbf{H}_{B I} \mathbf{v}_{k}$, where $\mathbf{\Psi}_t=diag(\Psi_t^1, \Psi_t^2,..., \Psi_t^l)$, the beamformer according to the refractive beams can be expressed as $\mathbf{q}_k^r=\mathbf{\Psi}_r \mathbf{H}_{B I} \mathbf{v}_{k}$ where $\mathbf{\Psi}_r=diag(\Psi_t^1, \Psi_r^2,..., \Psi_r^l)$.
Assume that there are two areas $p_t$ and $p_r$ symmetrical about the IOS, and
the equivalent channels between the IOS and two areas are $\mathbf{h}_t$ and $\mathbf{h}_r$, respectively.
The gain of the reflective beams corresponding to $p_t$ can be expressed as

\begin{small}
\vspace{-1em}
\begin{equation}\label{prop1}
\begin{aligned}
B(\mathbf{q}_{k}^t, \mathbf{h}_t)&=|\mathbf{h}_t\mathbf{q}_{k}^t|\\
&=|\mathbf{h}_t \mathbf{\Psi}_t\mathbf{H}_{B I} \mathbf{v}_{k}|\\
&=|\mathbf{h}_tdiag(\Psi_t^1, \Psi_t^2,..., \Psi_t^L) \mathbf{H}_{B I} \mathbf{v}_{k}|\\
&=|\mathbf{h}_tdiag(\Gamma_t^1 e^{j\phi_t^1}, \Gamma_t^2 e^{j\phi_t^2},..., \Gamma_t^L e^{j\phi_t^L}) \mathbf{H}_{B I} \mathbf{v}_{k}|.\\
\end{aligned}
\end{equation}
\end{small}

The gain of the refractive beams corresponding to $p_r$ can be expressed as

\begin{small}
\vspace{-1em}
\begin{equation}\label{prop2}
\begin{aligned}
B(\mathbf{q}_{k}^r, \mathbf{h}_r)&=|\mathbf{h}_r\mathbf{q}_{k}^t|\\
&=|\mathbf{h}_r \mathbf{\Psi}_r\mathbf{H}_{B I} \mathbf{v}_{k}|\\
&=|\mathbf{h}_rdiag(\Psi_r^1, \Psi_r^2,..., \Psi_r^L) \mathbf{H}_{B I} \mathbf{v}_{k}|\\
&=|\mathbf{h}_rdiag(\Gamma_r^1 e^{j\phi_r^1}, \Gamma_r^2 e^{j\phi_r^2},..., \Gamma_r^L e^{j\phi_r^L}) \mathbf{H}_{B I} \mathbf{v}_{k}|\\
&=|\mathbf{h}_rdiag(\Gamma_r^1 e^{j(\phi_t^1-c)},..., \Gamma_t^L e^{j(\phi_t^L-c)}) \mathbf{H}_{B I} \mathbf{v}_{k}|\\
&=|e^{-jc}\mathbf{h}_rdiag(\Gamma_r^1 e^{j\phi_t^1}, \Gamma_t^2 e^{j\phi_t^2},..., \Gamma_t^L e^{j\phi_t^L}) \mathbf{H}_{B I} \mathbf{v}_{k}|.\\
\end{aligned}
\end{equation}
\end{small}

\begin{itemize}
\item Far field: When $p_t$ and $p_r$ are in the far-field region, $\mathbf{h}_t=\gamma_{p_t}\mathbf{a}(\theta_{p_t},\psi_{p_t})$ and $\mathbf{h}_r=\gamma_{p_r}\mathbf{a}(\theta_{p_r},\psi_{p_r})$, where $\gamma$ represents the path loss.
We can find that $\frac{B(\mathbf{q}_{k}^t, \mathbf{h}_t)}{B(\mathbf{q}_{k}^r, \mathbf{h}_r)}=\frac{\gamma_{p_t}\Gamma_t}{\gamma_{p_r}\Gamma_r}$ if $\theta_{p_t}=\theta_{p_r}$ and $\psi_{p_t}=-\psi_{p_t}$.
\item Near field: When $p_t$ and $p_r$ are in the near-field region, $\mathbf{h}_t=[\sqrt{\frac{1}{4\pi (d_t^{1})^{2}}}e^{-j\frac{2\pi}{\lambda}d_t^{1}},...,\sqrt{\frac{1}{4\pi (d_t^{l})^{2}}}e^{-j\frac{2\pi}{\lambda}d_t^{l}}]$ and $\mathbf{h}_r=[\sqrt{\frac{1}{4\pi (d_r^{1})^{2}}}e^{-j\frac{2\pi}{\lambda}d_r^{1}},...,\sqrt{\frac{1}{4\pi (d_r^{l})^{2}}}e^{-j\frac{2\pi}{\lambda}d_r^{l}}]$, where $d_t^i=\sqrt{(x_t-x_i)^2+(y_t)^2+(z_t-z_i)^2}$ and $d_r^i=\sqrt{(x_r-x_i)^2+(y_r)^2+(z_r-z_i)^2}$.
We can derive that $\frac{B(\mathbf{q}_{k}^t, \mathbf{h}_t)}{B(\mathbf{q}_{k}^r, \mathbf{h}_r)}=\frac{\Gamma_t}{\Gamma_r}$ if $x_t=x_r$, $y_t=-y_r$, and $z_t=z_r$.
\end{itemize}
Based on~\eqref{prop1} and~\eqref{prop2}, we get that the reflective and refractive signals are symmetrical about the IOS in both near and far fields.
\vspace{-0.5em}
\end{proof}

Given the codeword $\mathbf{q}$, the digital precoder~$\mathbf{V}$ is initialized with $k$ same columns.
The IOS based analog beamformer can be calculated by $\mathbf{\Psi}_t=diag\left\{\Gamma_t e^{-j \angle \mathbf{q} \oslash\left(\mathbf{H}_{B I} \mathbf{v}_{k}\right)}\right\}$ and $\mathbf{\Psi}_r=diag\{\Gamma_r^1 e^{j(\phi_t^1-c)}, \Gamma_t^2 e^{j(\phi_t^2-c)},..., \Gamma_t^L e^{j(\phi_t^L-c)}\}$.
The codebook design can be summarized in \emph{Algorithm}~1.

\begin{algorithm}[t]
\label{alg:training}
%\footnotesize
\caption{Near-Far Field Codebook Design}
\LinesNumbered
\KwIn{The coverage $\mathscr{P}$ of the codeword $\mathbf{q}$;}
\KwOut{The codeword $\mathbf{q}$, the IOS based analog beamformer $\mathbf{\Psi}$ and the digital precoder $\mathbf{V}$ corresponding to $\mathbf{q}$;}
Construct the equivalent channel between the IOS and each area $p$ according to~\eqref{codeword}\;
Calculate the codeword  $\mathbf{q}$ via~\eqref{codeword}\;
Initialize the digital precoder $\mathbf{V}$ with $\mathbf{v}_1=..=\mathbf{v}_K$\;
Obtain the IOS based analog beamformer $\mathbf{\Psi}_t=diag\left\{\Gamma_t e^{-j \angle \mathbf{q} \oslash\left(\mathbf{H}_{B I} \mathbf{v}_{k}\right)}\right\}$ and $\mathbf{\Psi}_r=diag\{\Gamma_r^1 e^{j(\phi_t^1-c)}, \Gamma_t^2 e^{j(\phi_t^2-c)},..., \Gamma_t^L e^{j(\phi_t^L-c)}\}$.
\end{algorithm}
\vspace{-0.8em}
\subsection{Multi-User Beam Training Mechanism}
\vspace{-0.3em}
During the beam training, user $k$ receives the signal according to each codeword $\mathbf{q}_{n}$, and finally regards the codeword with the largest received power~$\gamma_k$ as the optimal codeword $\mathbf{q}_{k}^*$.
Traditionally, the codewords in the codebook $\mathcal{Q}$ are performed sequentially at the BS, requiring large spectrum and time resources for beam training.
To deal with this problem, we carry out the \emph{multi-user beam training mechanism} where the designed codebook is searched based on a binary tree and each codeword is generated to cover multiple areas~\cite{multibeam}.
As shown in Fig.~\ref{Fig:Hierarchical}, except for the bottom layer, each layer consists of two codewords each covering half of the total $P$ areas.
Now we dwell on the multi-user beam training process.

\begin{figure}[t]
\vspace{-1em}
\setlength{\abovecaptionskip}{0pt}
\setlength{\belowcaptionskip}{0pt}
	\centering
    \includegraphics[width=0.42\textwidth]{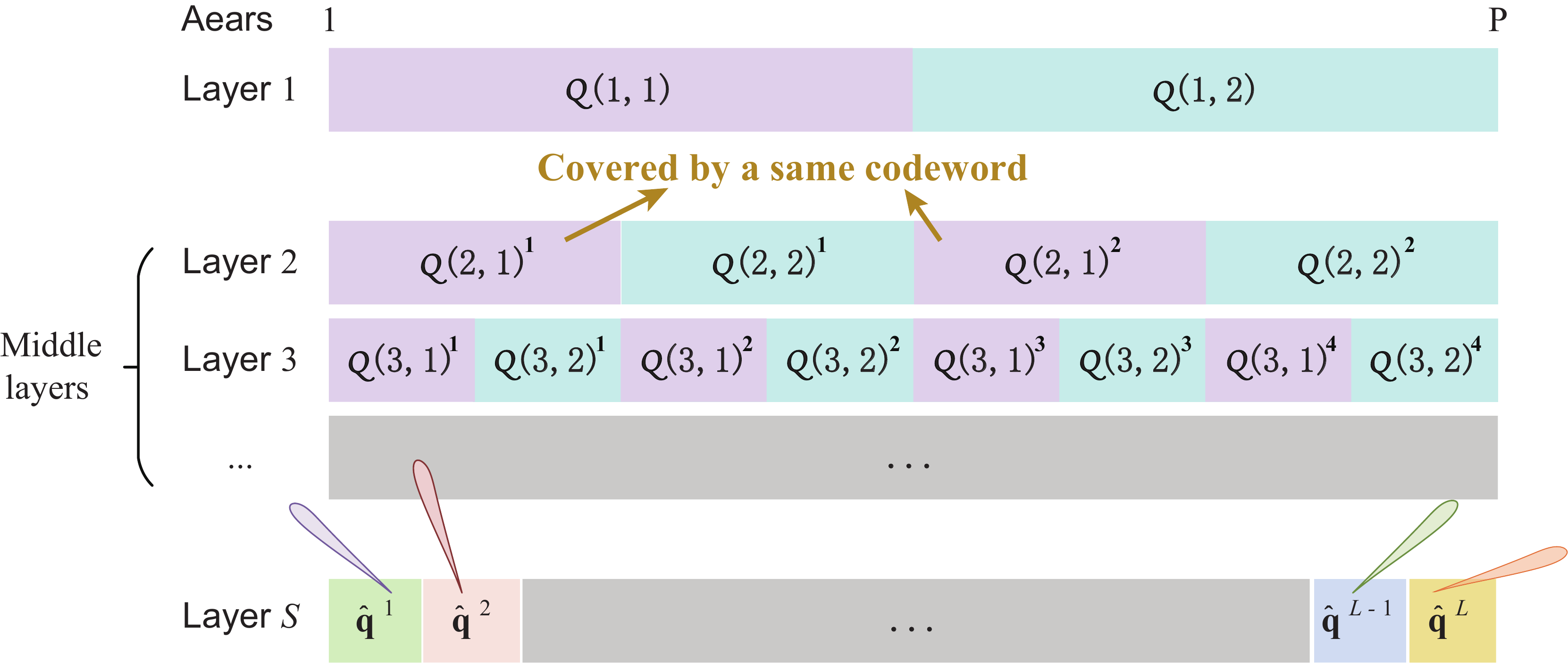}
	\caption{Hierarchical multi-user beam training mechanism.}
	\label{Fig:Hierarchical}
\end{figure}

\subsubsection{The upper layers}
In each layer $1\le s \le S-1$, we design two codewords $\mathcal{Q}(s,1)$ and $\mathcal{Q}(s,2)$ according to \emph{Algorithm} 1.
The beam training is performed for each user simultaneously in every layer, as described below.
\begin{itemize}
\item Partition of the areas: The coverage of each codeword in the $(s-1)$-th layer is evenly partitioned into two groups, i.e.,  $\mathscr{P}(s, 1)$ and $\mathscr{P}(s, 2)$.
For example, the coverage of the codewords $\mathcal{Q}(1,1)$ and $\mathcal{Q}(1,2)$ in the top layer is $\mathscr{P}(1,1)=   \left[1, \ldots, \frac{P}{2}\right]$  and  $\mathscr{P}(1,2)=\left[\frac{P}{2}+1, \ldots, P\right]$, respectively.
In the second layer, the coverage of $\mathcal{Q}(2,1)$ is $\mathscr{P}(2,1)= \left[1, \ldots, \frac{P}{4}\right] \cup\left[\frac{P}{2}+1, \ldots, \frac{3 P}{4}\right]$, consisting of part of the coverage of $\mathcal{Q}(1,1)$ and $\mathcal{Q}(1,2)$.
Similarly, the coverage of $\mathcal{Q}(2,2)$ is $\mathscr{P}(2,2)=\left[\frac{P}{4}+1, \ldots, \frac{P}{2}\right] \cup\left[\frac{3 P}{4}+1, \ldots, P\right]$.

\item Design the codewords: Given $\mathscr{P}(s, 1)$  and  $\mathscr{P}(s, 2)$, the codewords  $\mathcal{Q}(s, 1)$ and $\mathcal{Q}(s, 2)$ can be obtained by~\emph{Algorithm} 1.
\item Perform beam training: The codewords $\mathcal{Q}(s, 1)$ and $\mathcal{Q}(s, 2)$ are performed sequentially at the BS.
Each user $k$ receives the signal and records the index of the codeword with the larger received power $\tau_{s}^{k} \in\{1,2\}$.
The spatial region index $\beta_{s}^{k}$ of user $k$ in the $s$-th layer can be calculated by
\begin{equation}
\vspace{-0.3em}
\beta_{s}^{k}=2\beta_{s-1}^{k}+\tau_{s}^{k}-2.
\vspace{-0.3em}
\end{equation}
After training in the $(S-1)$-th layer, the results $\mathscr{B}_{S-1}=[\beta_{S-1}^1,...,\beta_{S-1}^K]$ are fed back to the BS.
\end{itemize}

\subsubsection{The bottom layer}
Different from the upper layers, beam training is performed for users one by one in the bottom layer.
Given the training results of the upper layers $\mathscr{B}_{S-1}$, the coverage of the optimal codeword for each user $k$ can be locked in two areas, expressed $\beta_{S}^{k} \in \{2\beta_{S-1}^{k}-1,2\beta_{S-1}^{k}\}$.
By performing these codewords sequentially, user $k$ obtains the optimal codeword and feeds its index $\beta_{S}^{k}$ back to the BS.
Then the optimal codeword $\mathbf{q}_{k}^{*}$ of user $k$ can be generated by~\eqref{codeword}.
After that, the IOS aided beamformer can be designed as $\mathbf{Q}$=$[\mathbf{q}_{1}^{*},...,\mathbf{q}_{K}^{*}]$.
\vspace{-0.4em}
\subsection{IOS aided Beamforming Scheme Design}
\vspace{-0.4em}
Given the training results, we configure the IOS based analog beamformer $\mathbf{\Psi}$ and the digital precoder $\mathbf{V}$ at the BS to solve the beamforming subproblem~\eqref{sub3}, where $f(\mathbf{v}_{k}, \mathbf{\Psi})$ can be expressed as
\begin{equation}
f(\mathbf{v}_{k}, \mathbf{\Psi})=a-2\text{Re}{\mathbf{\{b \psi\}}},
\vspace{-0.3em}
\end{equation}
where $\mathbf{\psi}$ consists of the diagonal elements of $\mathbf{\Psi}$.
The notation $a=\sum_{k}\left({(\mathbf{q}_{k}^{*}})^{H} \mathbf{q}_{k}^{*}+\mathbf{v}_{k}^{H} \mathbf{H}_{B I}^{H} \mathbf{H}_{B I} \mathbf{v}_{k}\right)$ and $\mathbf{b}=\sum_{k} (\mathbf{q}_{k}^{*})^{H} \operatorname{diag}\left\{\mathbf{H}_{B I} \mathbf{v}_{k}\right\} $.
\subsubsection{IOS based Analog Beamforming}
Without loss of generality, we assume that the energy split for the reflective and refractive signals is equal, i.e., $|\Gamma_t^l|=|\Gamma_r^l|$, and now focus on the design of the IOS phase.
Starting from a randomly initiated $\mathbf{\psi}$, the phase of each IOS element~$l$ is updated given other fixed elements iteratively.
Since the elements of Re$\{\mathbf{b \psi}\}$ are mutually independent, problem~\eqref{sub3} can be decoupled into~$L$ subproblems
\begin{equation}
\vspace{-0.3em}
\phi_t^l=\argmax\limits_{\phi_t^l} f(\phi_t^l),
\vspace{-0.3em}
\end{equation}
where $f(\phi_t^l)=b_l e^{j\phi_t^l}+b_l^{*} e^{-j\phi_t^l}$.
The notations~$b_l$ and~$\phi_t^l$ denote the $l$-th element of $\mathbf{b}$ and the IOS phase, respectively.
The phase of the $l$-th element can be calculated~by
\begin{equation}
\vspace{-0.3em}
\phi_t^l=\pi-\arctan\frac{\text{Im}\{b_l\}}{\text{Re}\{b_l\}},
\vspace{-0.3em}
\end{equation}
\subsubsection{Digital Beamforming at the BS}
The $k$-th column of the digital precoder $\mathbf{V}$ is considered as $\mathbf{v}_k=\widetilde{\mathbf{v}}_k p_k^{\frac{1}{2}}$, where $p_k$ is the transmit power allocated to the signal intended for the user $k$.
Given the IOS based analog beamformer $\mathbf{\Psi}$, $\widetilde{\mathbf{v}}_k$ can be calculated as
\begin{equation}
\vspace{-0.3em}
\tilde{\mathbf{v}}_{k}=\frac{\left[\left(\mathbf{\Psi} \mathbf{H}_{B I}\right)^{H} \mathbf{\Psi} \mathbf{H}_{B I}\right]^{-1}\left(\mathbf{\Psi} \mathbf{H}_{B I}\right)^{H} \mathbf{q}_{k}}{\left|\left[\left(\mathbf{\Psi H}_{B I}\right)^{H} \mathbf{\Psi} \mathbf{H}_{B I}\right]^{-1}\left(\mathbf{\Psi} \mathbf{H}_{B I}\right)^{H} \mathbf{q}_{k}\right|}.
\vspace{-0.3em}
\end{equation}

Therefore, the received signal for all users can be given by $\mathbf{y}=\mathbf{\widetilde{H} P x}+\mathbf{n}$, where $\mathbf{\widetilde{H}}$ and $\mathbf{P} = diag \{p_1,...,p_K\}$ denote the effective channel matrix and the transmit power matrix, respectively.
Each element of $\mathbf{\widetilde{H}}$ can be expressed by~$\widetilde{H}_{k,k'}=\mathbf{w}_k^H \mathbf{H}_{IU} \mathbf{\Psi} \widetilde{\mathbf{v}}_{k'}$.
Each $p_k$ can be obtained by water-filling, given by $p_k =\frac{1}{v_k}\max\{\frac{1}{\mu}-v_k\sigma^2,0\}$, where $v_k$ is the $k$-th diagonal element of $\widetilde{\mathbf{H}}^H\widetilde{\mathbf{H}}$ and $\mu$ is a normalized factor which is chosen to satisfy $\sum_{k} \max \left\{\frac{1}{\mu}-v_{k} \sigma^{2}, 0\right\}=P_{T}$~\cite{multibeam}.

\begin{figure*}[t]
\setlength{\abovecaptionskip}{0pt}
\setlength{\belowcaptionskip}{0pt}
\centering
\subfigure[Symmetrical beams of IOS.]
{
\label{Fig:beams}
\includegraphics[width=0.35\textwidth]{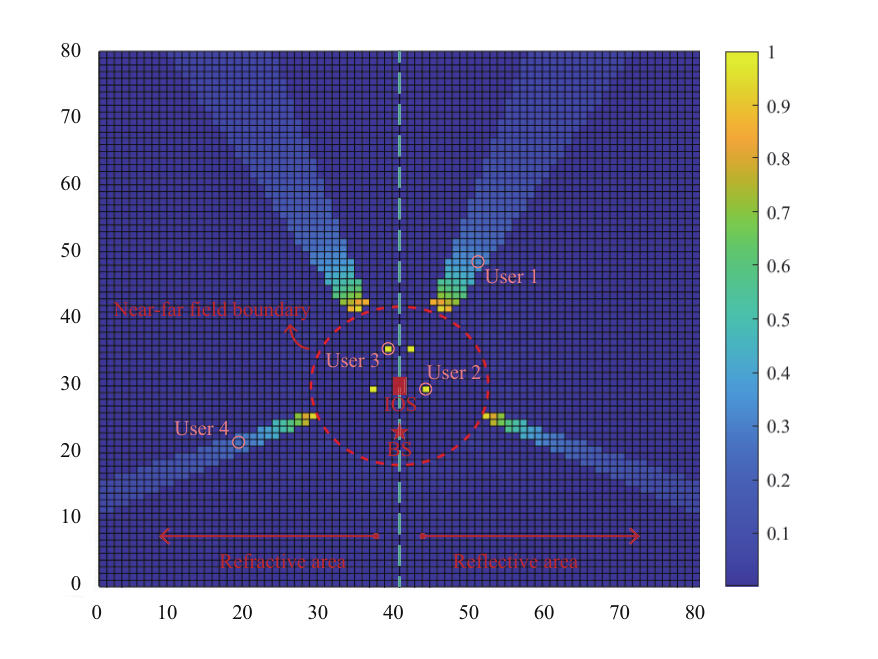}
}
\hspace{-3em}
\subfigure[Sum rate vs. SNR with different codebooks.]
{
\label{Fig:sumrate}
\includegraphics[width=0.35\textwidth]{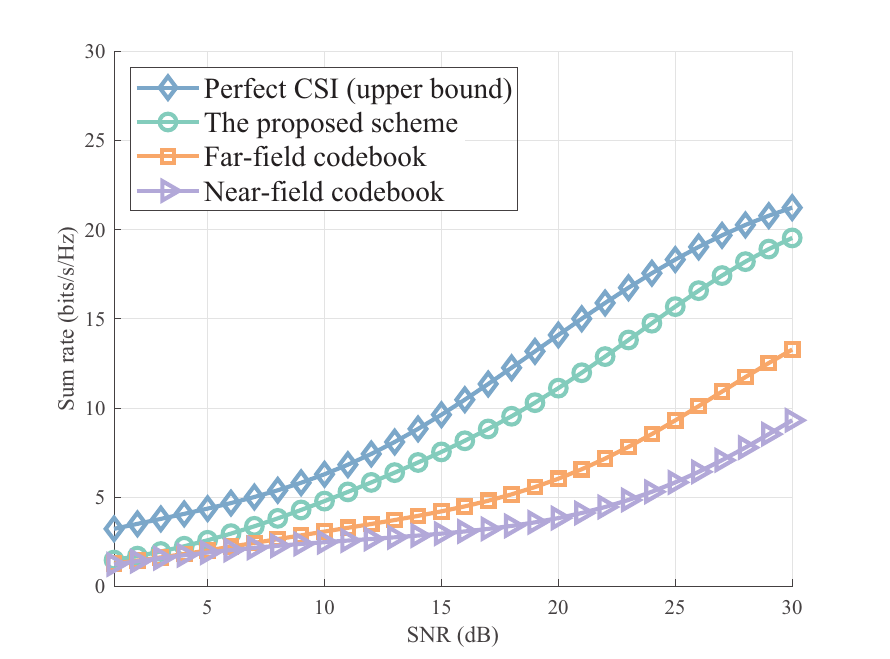}
}
\hspace{-3em}
\subfigure[Throughput vs. SNR.]
{
\label{Fig:compare}
\includegraphics[width=0.35\textwidth]{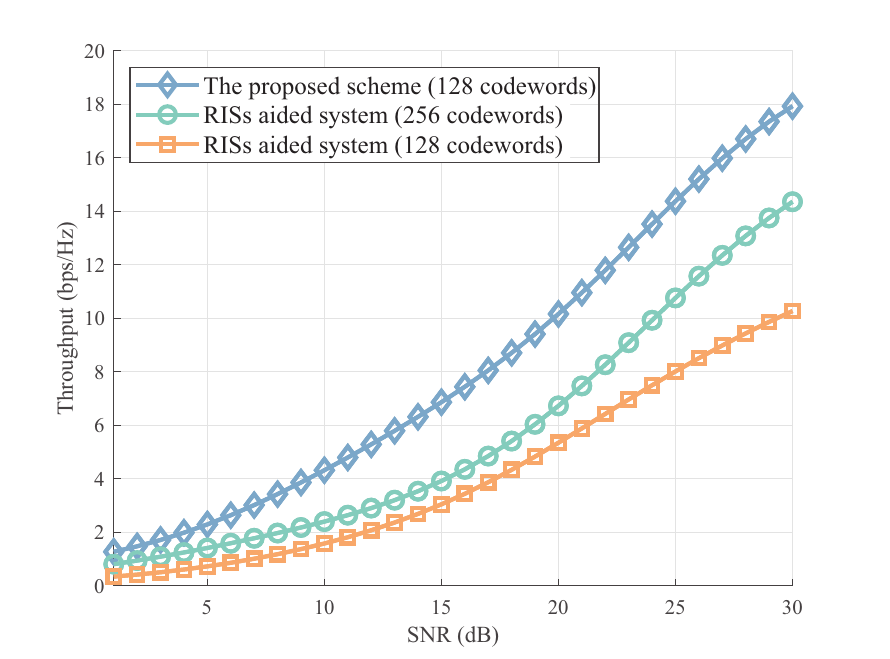}
}
\caption{Performance of the IOS aided system with the near-far field codebook and multi-user beam training.}
\vspace{-1.5em}
\end{figure*}
\vspace{-0.6em}
\section{Simulation results}
\vspace{-0.6em}
In this section, we first verify our theoretical results on symmetrical beams in~\emph{Proposition}~\ref{analysis} by simulations, and then evaluate our proposed scheme for the IOS aided system in terms of the sum rate and throughput.
We deploy a 1024-element IOS to aid a 32-antenna BS while each user is equipped with 4 antennas.
The carrier frequency is set~at~26~GHz.

As shown in Fig.~\ref{Fig:beams}, there are four users located in both near and far fields of the IOS, where user $1$ and $2$ are in the reflective area and user $3$ and $4$ are in the refractive area.
After the proposed multi-user beam training based on the near-far field codebook, the designed beams capture all users well, indicating the effectiveness of the near-far field codebook and proposed beam training scheme.
We can find that each reflective/refractive beam is generated with a beam symmetrical about the IOS in both near and far fields, which verifies~\emph{Proposition}~\ref{analysis} in Section~\ref{Codebook}.

In Fig.~\ref{Fig:sumrate}, we compare the sum rate of users with conventional codebooks.
We consider the perfect CSI case as an upper bound.
The conventional far/near-field codebook is designed by far-field channels or near-field channels with the same number of codewords in the near-far field codebook~\cite{hierarchical}.
We observe that the sum rate with different codebooks grows with the SNR.
Given the same codebook size, the sum rate of the proposed scheme is higher than that of the conventional near-field and far-field codebooks, and performs very close to the perfect CSI case.

In Fig.~\ref{Fig:compare}, two reconfigurable intelligent surfaces~(RISs), each of which has the same number of elements compared to the IOS, are placed together, serving reflective and refractive users respectively~\cite{RIS}.
Compared to RIS aided system with the same resolution of beam training, only half codewords are needed in IOS aided system benefiting from the symmetrical signals, which reduces the training overhead and brings a higher throughput in the IOS aided system.
When maintaining the same training overhead compared to the IOS aided system, the throughput of RIS aided system falls due to the low resolution of beam training which brings a descent in sum rate.
It implies by utilizing the symmetry of reflective and refractive waves, the system throughput can be~improved.
\vspace{-0.8em}
\section{Conclusions}
\vspace{-0.4em}
In this paper, we studied the codebook design and beam training for a multi-user system where an IOS is employed to achieve full-dimensional wireless communications by simultaneous signal reflection and refraction.
To fully utilize such a dual functionality, we analyzed the amplitude and phase coupling of the IOS elements theoretically, unveiling the symmetry relationship between the reflective and refractive signals of the IOS.
On this basis, we designed a near-far field codebook capable of serving users in both the near and far regions.
Based on the theoretical analysis and simulation results, we can draw the following conclusions:
1) The reflective and refractive signals have a symmetry relationship inherently in any case of reflective-refractive coupling.
2) The proposed scheme can capture users in both the near and the far fields well with low overhead and performs very close to the perfect CSI case.
3) Compared to the conventional ones, the proposed scheme can improve the sum rate and throughput.

\end{document}